\documentclass[US-letter,USenglish,authorcolumns,autoref]{lipics-v2019}
\usepackage[noend]{algpseudocode}
\usepackage{algorithm}
\usepackage{times}
\usepackage{amssymb}
\usepackage{amsmath}
\usepackage{latexsym}
\usepackage{graphics}
\usepackage{url}
\usepackage{verbatim}
\usepackage{color}
\usepackage{setspace}
\usepackage{multirow}
\usepackage{lineno}
\usepackage{booktabs}
\usepackage{graphicx}
\usepackage{caption}

\algnewcommand\And{\textbf{and}}
\algnewcommand\Or{\textbf{or}}
\algnewcommand\Not{\textbf{not}}
\algnewcommand\In{\textbf{in}}
\algnewcommand\Each{\textbf{each}}

\newcommand{\true}{\textsc{true}}
\newcommand{\false}{\textsc{false}}

\newcommand{\squishlist}{
 \begin{list}{$\bullet$}
  { \setlength{\itemsep}{0pt}
     \setlength{\parsep}{3pt}
     \setlength{\topsep}{3pt}
     \setlength{\partopsep}{0pt}
     \setlength{\leftmargin}{2.5em}
     \setlength{\labelwidth}{1em}
     \setlength{\labelsep}{0.5em} } }

\newcommand{\squishlisttwo}{
 \begin{list}{$\triangleright$}
  { \setlength{\itemsep}{0pt}
     \setlength{\parsep}{0pt}
    \setlength{\topsep}{0pt}
    \setlength{\partopsep}{0pt}
    \setlength{\leftmargin}{2em}
    \setlength{\labelwidth}{1.5em}
    \setlength{\labelsep}{0.5em} } }

\newcommand{\squishend}{
  \end{list}  }

\usepackage{listings}

\definecolor{verbgray}{gray}{0.9}

\lstnewenvironment{code}{%
  \lstset{backgroundcolor=\color{verbgray},
 frame=single,
  language=C,
  framerule=0pt,
  basicstyle=\ttfamily,
  showstringspaces=false,
  commentstyle=\color{blue}\textit,
  keywordstyle=\color{black}\bf,
  numbers=none,
  keepspaces=true,
  columns=fullflexible}}{}

\definecolor{shadecolor}{rgb}{.91, .91, .91}
\definecolor{bordercolor}{rgb}{.8, .8, .6}

\definecolor{ultramarine}{rgb}{0, 0.125, 0.376}

 \definecolor{arsenic}{rgb}{0.23, 0.27, 0.29}
 \definecolor{beige}{rgb}{0.96, 0.96, 0.86}

\definecolor{amber}{rgb}{1.0, 0.75, 0.0}
\definecolor{orange}{rgb}{1.0, 0.49, 0.0}
\definecolor{dandelion}{rgb}{0.94, 0.88, 0.19}

  \definecolor{indiagreen}{rgb}{0.07, 0.53, 0.03}
  \definecolor{huntergreen}{rgb}{0.21, 0.37, 0.23}

\newcommand{\blue}[1] {\textcolor{blue}{#1}}
\newcommand{\red}[1] {\textcolor{red}{#1}}

\newcommand{\bblue}[1] {{\bf \textcolor{blue}{#1}}}

\newcommand{\defo}[1] {\emph{\textcolor{blue}{#1}}}

\definecolor{shadecolor}{rgb}{.9, .9, .9}

 \usepackage{framed}

    {\endMakeFramed}

    \newenvironment{frshaded*}{%
    \MakeFramed {\advance\hsize-\width \FrameRestore}}%
    {\endMakeFramed}
    
    \newcounter{examplecounter}
\newenvironment{exam}{
 \begin{frshaded*}
    \refstepcounter{examplecounter}%
    \noindent
  \textbf{Example \arabic{examplecounter}}%
  \quad
}{%
\end{frshaded*}
}

{\endMakeFramed}

\newenvironment{frshaded2*}{%
    \MakeFramed {\advance\hsize-\width \FrameRestore}}%
    {\endMakeFramed}

\newenvironment{result}{
 \begin{frshaded2*}
}{%
\end{frshaded2*}

}

{\endMakeFramed}

\newenvironment{frshaded3*}{%
    \MakeFramed {\advance\hsize-\width \FrameRestore}}%
    {\endMakeFramed}

\definecolor{winered}{rgb}{0.5,0.2,0}
\usepackage{hyperref}
\hypersetup{ 
    colorlinks = true,
    allcolors={winered},
}

\newcommand{\edit}[1]{\textcolor{black}{#1}}

\usepackage[normalem]{ulem}

\DeclareMathOperator{\per}{per}

\title{Cut-Down de Bruijn Sequences}
\titlerunning{Cut-Down de Bruijn Sequences}
\author{Ben Cameron}{The King's University, Canada}{}{}{}
\author{Aysu G\"undo\u{g}an}{ University of Guelph, Canada}{}{}{}
\author{Joe Sawada}{University of Guelph, Canada}{}{}{}
\author{~}{~}{~}{}{}

\authorrunning{B. Cameron, A. G\"undo\u{g}an, and J. Sawada} 
\authorrunning{~} 

\Copyright{A. G\"undo\u{g}an, J. Sawada, and B. Cameron}
\Copyright{~}

\nolinenumbers

\begin{document}

\maketitle


\begin{abstract}
%
A cut-down de Bruijn sequence is a cyclic string of length $L$, where $1 \leq L \leq k^n$,  such that every substring of length $n$ appears \emph{at most} once.   
Etzion~[\emph{Theor. Comp. Sci} 44 (1986)] introduced an algorithm to construct binary cut-down de Bruijn sequences requiring $o(n)$ simple $n$-bit operations per symbol generated.  In this paper, we 
simplify the algorithm and improve the running time to $\mathcal{O}(n)$ time per symbol generated using $\mathcal{O}(n)$ space. Additionally, we develop the first successor-rule approach for constructing a binary cut-down de Bruijn sequence by leveraging recent ranking algorithms for fixed-density Lyndon words. Finally, we develop an algorithm to generate cut-down de Bruijn sequences for $k>2$ that runs in $\mathcal{O}(n)$ time per symbol using $\mathcal{O}(n)$ space after some initialization. 
\end{abstract}



\section{Introduction}
A \defo{de Bruijn sequence} (DB sequence) of span $n$, over an alphabet of size $k$, is a cyclic sequence of length $k^n$ such that every $k$-ary string of length $n$ appears as a substring exactly once.  For example, the following is a DB sequence for $n=6$ and $k=2$:
\begin{equation} \label{eq:DB6} [~0000001111110111100111000110110100110000101110101100101010001001~].
\end{equation}
\edit{
The \defo{de Bruijn graph}  of span $n$, over an alphabet of size $k$, is the directed graph $G(n,k) = (V,E)$  where $V$ is the set of all $k$-ary strings of length $n$ and there is a directed edge $e=(u,v) \in E$  from $u = u_1u_2\cdots u_n$ to $v=v_1v_2\cdots v_n$ if $u_2\cdots u_n = v_1\cdots v_{n-1}$. Each edge $e$ is labeled by $v_n$.  In this paper, the term \defo{cycle} corresponds to a sequence of edge labels obtained by traversing some cycle/circuit in the de Bruijn graph (or a related edge-labeled graph), and the notation $[\alpha]$ denotes the sequence $\alpha$ is cyclic.  For example,  Figure~\ref{fig:dbgraph} illustrates $G(3,2)$ 
and the cycles $[01]$ and $[1101100001]$.
 It is well known that a DB sequence of span $n$ is in one-to-one correspondence with an Euler cycle in $G(n{-}1,k)$. } 

\begin{figure}[ht]
    \centering
    \resizebox{4.5in}{!}{\includegraphics{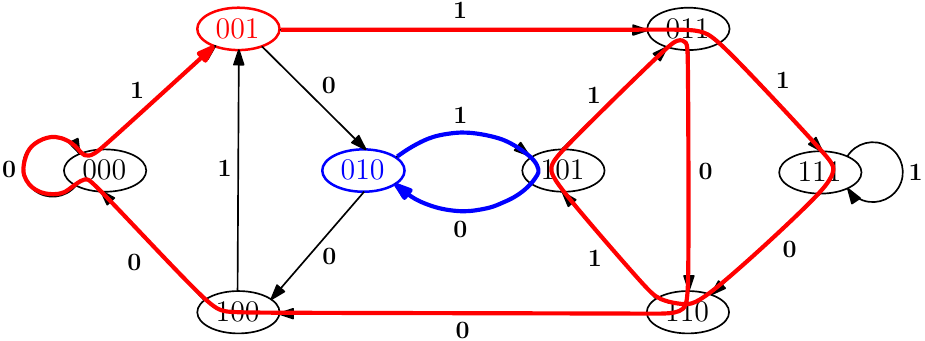}}
    \caption{The de Bruijn graph $G(3,2)$ highlighting cut-down DB sequences $[10]$ (blue) and $[1101100001]$  (red) of length two and ten, respectively. }
    \label{fig:dbgraph}
\end{figure}

For some applications it may be more convenient to produce a cycle of arbitrary length such that there are no repeated length-$n$ substrings, i.e., a cycle of arbitrary length in $G(n-1,k)$.  For instance, it may be more natural to consider the de Bruijn card trick~\cite{magic} using 52 cards rather than 32.  Also, for applications in robotic vision and location detection~\cite{magic,robot,sinden}, instead of forcing a location map to have length $k^n$, an arbitrary length allows for more flexibility.  
This gives rise to the notion of a \defo{cut-down de Bruijn sequence}\footnote{The term \emph{cutting-down}  is perhaps first used in~\cite{JSH09} to describe such sequences.} (cut-down DB sequence),  which is a cyclic sequence of length $L$ over an alphabet of size $k$, where $1 \leq L \leq k^n$,  such that every substring of length $n$ appears \emph{at most} once.    As an example, the following is a binary cut-down DB sequence of length 52:
$$ [~0000001111001110001101101001100001011101011001010001~].$$
Note that every substring of length six, including in the wraparound, appears at most once.
Any cut-down DB sequence of length $L$ with respect to $k$ and $n$ is also a cut-down DB sequence with respect to $k$ and $n+1$~\cite{etzion86}.   Thus, throughout this paper we assume  $k^{n-1} < L \leq k^{n}$. 

 
Cut-down DB sequences are known to exist for all lengths $L$ and any alphabet of size $k$~\cite{lempel71} (for $k{=}2$ see~\cite{Yoeli}).   In bioinformatics, the alphabet $\{C,G,A,T\}$ of size $k=4$ is of particular interest and there are a number of applications that apply DB sequences and their relatives~\cite{neural,pevzner}.   
A simple algorithm to construct binary cut-down DB sequences based on linear feedback shift registers and primitive polynomials is given by Golomb~\cite[P.~193]{golomb2017}; it runs in $\mathcal{O}(n)$-amortized time per symbol using $\mathcal{O}(n)$ space. However, the construction has an exponential-time delay before producing the first symbol, requires a specific primitive polynomial for each order $n$, and there is no way to determine if a given length-$n$ string appears as a substring without generating the entire cycle.  This approach is generalized to construct cut-down DB sequences where $k$ is a prime power, and subsequently extended to handle arbitrary sized alphabets by applying additional number theoretic results~\cite{landsberg}.  Although no formal algorithmic analysis is provided, the formulation appears to share properties no better than the related binary construction. An algebraic approach for when $k$ is a prime power is also known~\cite{hemmati}.
A cycle-joining based approach
 to construct cut-down DB sequences was developed by Etzion~\cite{etzion86} for $k=2$; it requires $o(n)$ simple $n$-bit operations to generate each symbol.  The approach follows two main steps:
\begin{itemize}
    \item First, an initial cycle is constructed with length $L+s$, where $0 \leq s < n$, using the well-known cycle-joining approach. 
    \item Second, depending on $s$, up to $\lceil \log n \rceil$ small cycles are detected and removed to obtain a cycle of length $L$.
\end{itemize}
The resulting algorithm can construct an exponential number of cut-down DB sequences for any given $L$; however, their algorithm is not optimized to generate a single cut-down DB sequence.  
Etzion's construction also has a downside for some applications: It starts with a specific length-$n$ string and the historical context matters to produce successive symbols. This means the resulting cycle does not have a corresponding successor rule, and  furthermore, testing whether or not a specific string belongs to the cycle may involve generating the entire cycle.  

The main results of this paper are as follows: 
%
\vspace{-0.1in}

\begin{enumerate}
    \item We simplify Etzion's approach and develop an algorithm to construct a binary cut-down DB sequence in $\mathcal{O}(n)$ time per symbol using $\mathcal{O}(n)$ space. \smallskip

    \item We develop the first successor-rule approach to construct a binary cut-down DB sequence using $\mathcal{O}(n^{1.5})$-amortized simple operations on $n$-bit numbers per symbol generated, and polynomial space.  The algorithm can start with any string on the cycle and the context does not matter when producing successive symbols. \edit{Determining whether or not a length-$n$ string appears as a substring on the sequence can be determined in $\mathcal{O}(n^3)$ time.}  \smallskip
    
    \item We develop an algorithm to generate cut-down DB sequences for $k>2$ that runs in $\mathcal{O}(n)$ time per symbol using $\mathcal{O}(n)$ space after some initialization requiring polynomial time and space.  A number of non-trivial adaptations to the binary algorithm are required to generalize to larger alphabets. 
\end{enumerate}  \vspace{-0.1in}

\smallskip

\noindent
All three algorithms require polynomial space and generate each symbol with polynomial-time delay.

\noindent


\medskip

\noindent
{\bf Related work.}
A \defo{generalized de Bruijn sequence}, as defined in~\cite{blum-conf}, is a cut-down DB sequence of length $k^{n-1} < L \leq k^n$ with an additional property: \emph{every $k$-ary string of length $n{-}1$ appears as a substring}.  Their existence is known for all $L$ and $k$~\cite{shallit}. For special values of $L$, these sequences can be generated by considering the base $k$  expansion of $1/L$~\cite{blum-conf,blum,matthews}.  An  algorithm based on Lempel's $D$-morphism~\cite{lempel} has recently been proposed to construct these sequences~\cite{abhi} that have an even stronger property:
    \emph{every $k$-ary string of length $j\leq L$ appears either $\lfloor L/k^j \rfloor$ or  $\lceil L/k^j \rceil$ times as a substring}. 
    We call sequences with this latter property
    \defo{balanced cut-down de Bruijn sequences}.
    The proposed algorithm can generate the sequences in $\mathcal{O}(1)$-amortized time per symbol, but it requires  exponential space and there is an exponential time delay before outputting the first symbol.


\edit{
\defo{Repeat-free sequences} have all of the properties of cut-down DB sequences except the cyclic property; they are  prefixes of a DB sequence and correspond to paths in the de Bruijn graph. They were considered from an algorithmic perspective in~\cite{Yaakobi2023,Yaakobi2021} and discussed from a combinatorial perspective under the name \defo{partial de Bruijn $\ell$-sequences} in~\cite{Zulingetal2017}. 
}

\smallskip

\noindent
{\bf Outline of paper.}
In Section~\ref{sec:background}, we provide some background on the cycle-joining method and 
a simple successor rule to construct  DB sequences. 
In Section~\ref{sec:etzion}, we review Etzion's approach for constructing binary cut-down DB sequences. In Section~\ref{sec:binary}, we present in detail our simplified algorithm to construct binary cut-down DB sequences; in Section~\ref{sec:successor}, we present the first successor-rule algorithm for constructing binary cut-down DB sequences.  In Section~\ref{sec:kary} 
we extend our binary algorithm to work for $k >2$.
Implementation of our algorithms, written in C, are available for download at~\url{http://debruijnsequence.org/db/cutdown}~\cite{dborg}; this resource also provides a comprehensive background on DB sequences and their constructions.

\section{Background} \label{sec:background}

Let $\Sigma$ denote the alphabet $\{0,1, \ldots, k{-}1\}$ where $k \geq 2$.  Let $\Sigma^n$ denote the set of all length-$n$ strings over $\Sigma$. Let $\alpha = a_1a_2\cdots a_n$ be a string in $\Sigma^n$. 
Let $\alpha^t$ denote $t$ copies of $\alpha$ concatenated together.  The \defo{period} of $\alpha$, denoted $\per(\alpha)$, is the smallest integer $p$ such that $\alpha = (a_1\cdots a_p)^t$ for some $t>0$.  If $\alpha$ has period less than $n$ it is said to be \defo{periodic}; otherwise it is \defo{aperiodic}. The lexicographically smallest element in an equivalence class of words under rotation is called a \defo{necklace}.

A \defo{Lyndon word} is an aperiodic necklace.   The \defo{weight} of a string is the sum of its elements (when $k=2$, weight is sometimes referred to as \defo{density}). Let $T_k(n,w)$ denote the number of $k$-ary strings of length $n$ and weight $w$. Note $T_2(n,w) = {n \choose w}$. \edit{Let $L_k(n,w)$ denote the number of $k$-ary Lyndon words of length $n$ and weight $w$.  By partitioning the strings of $T_k(n,w)$ into equivalence classes under rotation and considering the period of the string in each class (see Example~\ref{exam:PCR}) observe that $T_k(n,w) = \sum_{d\vert n} \frac{n}{d}L(\frac{n}{d},\frac{w}{d})$.  By applying Mobi\"{u}s inversion we have:
%
$$L_k(n,w) = \frac{1}{n}\sum\limits_{d\vert \gcd(n,w)} \mu(d)\ T_k\left(\frac{n}{d}, \frac{w}{d}\right),$$
where $\mu$  is the Mobi\"{u}s function.  When $k=2$, the formula is derived in~\cite{GilbertRiordan1961}
and applied in~\cite{hartman}.} 


A \defo{feedback function} is a function $f: \Sigma^n \rightarrow \Sigma$.  A \defo{feedback shift register} (FSR) is a function $F: \Sigma^n \rightarrow \Sigma^n$ defined as $F(\alpha) = a_2a_3\cdots a_nf(\alpha)$, given a feedback function $f$.  An FSR is said to be \defo{nonsingular} if it is one-to-one.  
The \defo{pure cycling register} (PCR) is the FSR with feedback function $f(\alpha) = a_1$.  It partitions $\Sigma^n$ into an equivalence class of strings under rotation.  Thus, the cycles induced by the PCR, called \defo{PCR cycles}, are in one-to-one correspondence with the necklaces of order $n$ and also with Lyndon words whose lengths divide $n$. They also appear as cycles in $G(n-1,k)$.  Recall, we use the notation $[\alpha]$ to denote a cycle.  

\begin{exam} \label{exam:PCR} \small
 Let $\Sigma = \{0,1\}$ and let $n=6$.  The following are the 14 equivalence classes of $\Sigma^6$ under rotation, where the first string in each class is a necklace.  
 \begin{center}
 \begin{tabular}{ccccccc}

\ \bblue{0}00000  \  & \  \bblue{0}00001  \ &  \ \bblue{0}00011   \ & \ \bblue{0}00101   \ & \  \bblue{0}00111 \   &  \ \bblue{0}01001 \  & \ \bblue{0}01011 \   \\
       & \bblue{0}00010 & \bblue{0}00110 & \bblue{0}01010 & \bblue{0}01110 & \bblue{0}10010 & \bblue{0}10110  \\
       & \bblue{0}00100 & \bblue{0}01100 & \bblue{0}10100 & \bblue{0}11100 & \bblue{1}00100 & \bblue{1}01100  \\
       & \bblue{0}01000 & \bblue{0}11000 & \bblue{1}01000 & \bblue{1}11000 &       & \bblue{0}11001  \\
       & \bblue{0}10000 & \bblue{1}10000 & \bblue{0}10001 & \bblue{1}10001 &       & \bblue{1}10010  \\
       & \bblue{1}00000 & \bblue{1}00001 & \bblue{1}00010 & \bblue{1}00011 &       & \bblue{1}00101  \\
       & \\ 
       & \\
\bblue{0}01101 & \bblue{0}01111 & \bblue{0}10101 & \bblue{0}10111 & \bblue{0}11011 & \bblue{0}11111 & \bblue{1}11111 \\  
\bblue{0}11010 & \bblue{0}11110 & \bblue{1}01010 & \bblue{1}01110 & \bblue{1}10110 & \bblue{1}11110 &  \\
\bblue{1}10100 & \bblue{1}11100 &        & \bblue{0}11101 & \bblue{1}01101 & \bblue{1}11101 &  \\
\bblue{1}01001 & \bblue{1}11001 &        & \bblue{1}11010 &        & \bblue{1}11011 &  \\
\bblue{0}10011 & \bblue{1}10011 &        & \bblue{1}10101 &        & \bblue{1}10111 &  \\
\bblue{1}00110 & \bblue{1}00111 &        & \bblue{1}01011 &        & \bblue{1}01111 & 
 \end{tabular}
  \end{center}
 
 \noindent
The following 14 PCR cycles are in one-to-one correspondence with the set of Lyndon words of lengths 1, 2, 3, and 6 (lengths that divide $n=6$):
\begin{center}
\begin{tabular}{ccccccc}
    $[0]$ &  $[000001]$ & $[000011]$ & $[000101]$ & $[000111]$ & $[001]$ & $[001011]$ \\
    $[001101]$ &  $[001111]$ & $[01]$ & $[010111]$ & $[011]$ &  $[011111]$ & $[1]$.
\end{tabular}
\end{center}

\vspace{-0.05in}
   
\end{exam}

\noindent
We say a cycle \defo{contains} a string $\alpha$ if $\alpha$ has length $n$ and is found as a substring on the cycle; we say $\alpha$ \defo{belongs} to the cycle.  Note the strings belonging to a PCR cycle all have the same weight.  Thus,
let the weight of a PCR cycle be the weight of its corresponding length $n$ necklace.  For example, when $n=6$ the weight of $[000001]$ is one and the weight of $[01]$ is three since its corresponding necklace is 010101. 

A \defo{universal cycle} for a set $\mathbf{S}$ of length-$n$ strings is a cyclic sequence of length $|\mathbf{S}|$ such that every string in $\mathbf{S}$ appears as a substring exactly once. 
\edit{Two universal cycles for $\mathbf{S_1}$ and $\mathbf{S_2}$  are said to be \defo{disjoint} if the sets $\mathbf{S_1}$ and $\mathbf{S_2}$ are disjoint.}
Of course, a DB sequence is a special case of a universal cycle when $\mathbf{S}$ corresponds to all $k$-ary strings of length $n$.  A cut-down DB sequence of length $L$ also corresponds to a universal cycle of length $L$; however, the corresponding set $\mathbf{S}$ is not necessarily known \emph{a priori}. 
A \defo{UC-successor} for $\mathbf{S}$  is a feedback function whose corresponding FSR can be repeatedly applied to construct a universal cycle for $\mathbf{S}$ starting from any string in $\mathbf{S}$ (see the upcoming Algorithm~\ref{algo:MC} for a specific example).  When 
$\mathbf{S} = \Sigma^n$ a UC-successor is said to be a \defo{de Bruijn-successor}.

\subsection{Cycle joining}
   
One of the most common ways to construct a DB sequence is by applying the  \defo{cycle-joining method}~\cite{golomb2017}, which is akin to Hierholzer's method for finding Euler cycles in graphs~\cite{hierholzer}.  This approach repeatedly joins pairs of disjoint (universal) cycles that share a node $v=a_2\cdots a_n$ in $G(n-1,k)$.  The two cycles are said to be joined via  a \defo{conjugate pair} ($xa_2\cdots a_n$, $ya_2\cdots a_n$), where $x$ and $y$ correspond to the labels of incoming edges to $v$ for each cycle. This leads to the following lemma which is implicitly applied in all cycle-joining constructions and formalized for $k=2$ in~\cite{dbrange}.

\begin{lemma} \label{thm:concat}
Let $\mathbf{S}_1$ and $\mathbf{S}_2$ be disjoint subsets of $\Sigma^n$
such that $xa_2\cdots a_n \in \mathbf{S}_1$ and $ya_2\cdots a_n \in \mathbf{S}_2$; $(xa_2\cdots a_n, ya_2\cdots a_n)$ is a conjugate pair.
If $U_1$ is a universal cycle for $\mathbf{S}_1$ and $U_2$ is a universal cycle for $\mathbf{S}_2$, each with prefix $a_2\cdots a_n$, 
then $U = U_1U_2$ is a universal cycle for $\mathbf{S}_1 \cup \mathbf{S}_2$.
\end{lemma}

When the initial cycles are those induced by an underlying nonsingular FSR, the joining of the cycles can be viewed as a tree in the binary case.  As an example,  Figure~\ref{fig:fulltree} illustrates the PCR cycles for $n=6$ and $k=2$, and one way they can be joined together to create a DB sequence, or equivalently, an Euler cycle in $G(5,2)$.  When extending this idea to larger alphabet sizes, the tree visualization no longer applies in general (see~\cite{karyframework} and the upcoming Figure~\ref{fig:T43}).

\begin{figure}[ht]
    \centering
    \resizebox{!}{4.2in}{\includegraphics{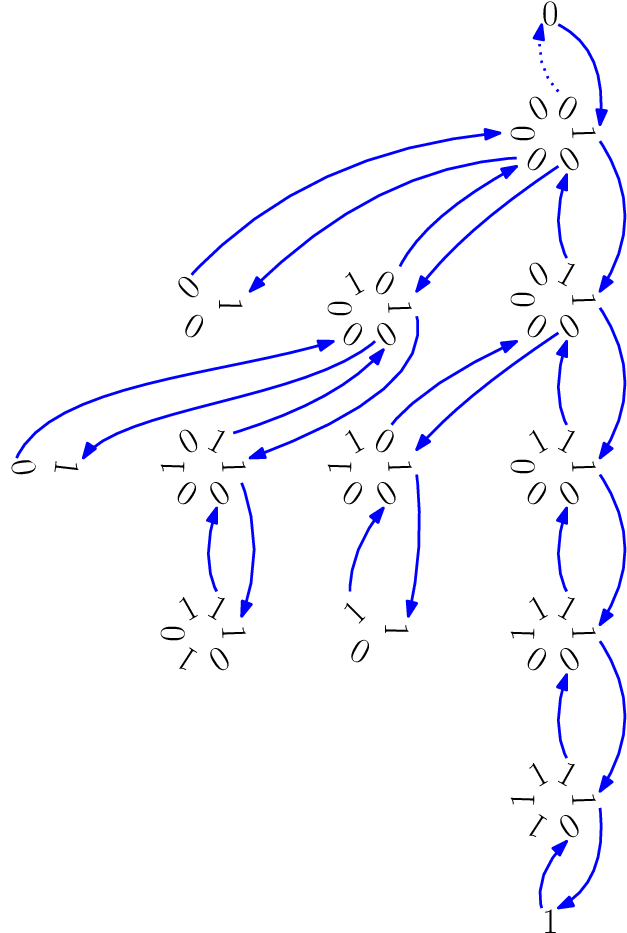}}
    \caption{The 14 PCR cycles for $n=6$ joined by applying the de Bruijn successor {\sc PCR3}. }
    \label{fig:fulltree}
\end{figure}

A general framework based on the cycle-joining approach leads to many simple UC-successors~\cite{binaryframework}.  Application of this framework rediscovers many previously known DB sequence constructions including one by Jansen~\cite{jansen} that was revisited in~\cite{SWW16} with respect to the PCR.  This particular construction starts with $[0]$ as the root cycle and repeatedly joins PCR cycles by increasing weight via  conjugate pairs $(1a_2\cdots a_n, 0a_2\cdots a_n)$, where $a_2\cdots a_n1$ is the necklace representative of the new PCR cycle begin joined. 
This construction is illustrated in Figure~\ref{fig:fulltree}, where the symbol pointed to by a downward edge is the ``last symbol'' in the corresponding cycle's necklace representative.
Thus, given a necklace $\alpha = a_1\cdots a_n \neq 0^n$, the parent cycle of $[\alpha]$ in the ``cycle-joining tree''
rooted at $[0]$ is $[a_1\cdots a_{n-1}0]$; the conjugate pair that joins them is $(0a_1\cdots a_{n-1}, 1a_1\cdots a_{n-1})$.
The resulting de Bruijn successor (below),  labeled {\sc PCR3} in~\cite{binaryframework} and  the ``Granny'' in~\cite{concattree},  is perhaps the simplest of all de Bruijn successors.  Note 
$\overline{x}$ denotes the complement of the bit $x$.

\begin{result}
\noindent {\bf PCR3 de Bruijn successor}:

\medskip

{\sc PCR3}$(\alpha) = \left\{ \begin{array}{ll} \overline{a}_1 &\ \ \mbox{if $\color{blue}{a_2a_3\cdots a_n1}$ is a necklace;}\\ {a_1} \ &\ \ \mbox{otherwise.}\end{array} \right.$ 
\end{result}

\noindent
When the FSR with feedback function {\sc PCR3} is repeatedly applied to the starting string $000000$ for $n=6$, as illustrated in  Figure~\ref{fig:fulltree}, it produces the  DB sequence in (\ref{eq:DB6}),
%
 where the first bit of the current string $\alpha$ is output before each application of the rule (see the upcoming Algorithm~\ref{algo:MC}).  The arcs between cycles in  Figure~\ref{fig:fulltree} correspond to the cases when {\sc PCR3} returns $\overline{a}_1$. 

%

Of course, the cycle-joining process yielding {\sc PCR3} can be applied to any subset of PCR cycles as long as they are ``connected'' via the defined conjugate pairs, i.e., the PCR cycles form a subtree of the complete cycle joining tree induced by {\sc PCR3}\footnote{See a related discussion in~\cite{concattree}.}.  If $\mathbf{S}$ is the set of strings belonging to some PCR cycle in such a subset (subtree), then let \blue{$\mathcal{S}$} denote the set of all possible sets $\mathbf{S}$.  In particular, we will be interested in a set $\mathbf{S} \in \mathcal{S}$ obtained from the PCR cycles with weight less than some $m>0$ together with a subset of PCR cycles with weight $m$ (see Example~\ref{exam:MC}).  A UC-successor for any $\mathbf{S} \in \mathcal{S}$ can be obtained from the PCR3 de Bruijn successor by additionally ensuring that {\sc PCR3}($a_1a_2\cdots a_n$) maps to $\overline{a}_1$ only if $a_2\cdots a_n\overline{a}_1$ is in $\mathbf{S}$, i.e., it does not attempt to join a cycle outside the specific subset.  

\begin{result}
\noindent {\bf PCR3 UC successor for $\mathbf{S} \in \mathcal{S}$}:

\medskip

{\sc PCR3$'$}$(\alpha) = 
\left\{ 
   \begin{array}{ll} 
     \overline{a}_1 &\ \ \mbox{if $a_2a_3\cdots a_n1$ is a necklace \color{blue}{and $a_2a_3\cdots a_n\overline{a}_1 \in \mathbf{S}$;}}\\ {a_1} \ &\ \ \mbox{otherwise.}
  \end{array} 
\right.$

\end{result}

Observe that {\sc PCR3} is just a special case of {\sc PCR3}$'$ when $\mathbf{S}$ is the set of all binary strings of length $n$.  Note that if $a_1\cdots a_n$ is in $\mathbf{S}$, then so is $a_2\cdots a_na_1$; they belong to the same PCR cycle. Starting with any string $\alpha \in \mathbf{S}$, Algorithm~\ref{algo:MC} applies this UC-successor to construct a universal cycle for $\mathbf{S}$, applying the original definition for {\sc PCR3}.  Later, we will specify further implementation details for a specific $\mathbf{S}$.

\begin{algorithm}[th]           
\caption{Pseudocode for constructing a universal cycle for $\mathbf{S} \in \mathcal{S}$ assuming  $\alpha = a_1a_2\cdots a_n \in \mathbf{S}$.} \label{algo:MC}   
        
\begin{algorithmic} [1]                   
\Procedure{UC}{$\alpha$}
\For{$i\gets 1$ {\bf to} $|\mathbf{S}|$}
    \State \Call{Print}{$a_1$}

    \State $x \gets$ \Call{PCR3}{$\alpha$}
    \State $\beta \gets a_2\cdots a_n x$
     \If{$\beta \notin \mathbf{S}$ } \ \  $x \gets \overline{x}$ \EndIf

    \State $\alpha \gets a_2\cdots a_n x$ 
  
\EndFor

\EndProcedure
\end{algorithmic}
\end{algorithm} 

\begin{lemma}
Algorithm~\ref{algo:MC} generates a universal cycle for $\mathbf{S}$, where $\mathbf{S} \in \mathcal{S}$.
\end{lemma}

This DB sequence constructed by PCR3 has an important property not shared by the other simple feedback functions 
presented in~\cite{binaryframework}: The strings belonging to  $\mathbf{Z}_i$ (defined in Section~\ref{sec:cut}) appear contiguously as substrings in the corresponding DB sequence.  

\section{Etzion's approach} \label{sec:etzion}
In this section, we outline Etzion's~\cite{etzion86} approach for constructing a binary cut-down DB sequence.   Recall that $L$ is the length of the cut-down DB sequence and $2^{n-1} < L \leq 2^n$.
The two primary steps in Etzion's construction are as follows, where the \defo{surplus} $s$ is an integer in $\{0,1,\ldots , n{-}1\}$:  \smallskip
\begin{enumerate}
    \item Construct a Main Cycle (MC) that has length $L+s$.
    \item Cut out up to $\lceil \log s \rceil$ small cycles from the MC to yield a cycle of the desired length $L$.
\end{enumerate}

To construct an MC, a subset of the PCR cycles are selected based on their weight and period.  Enumeration of strings by weight and period determine which cycles to include.  Considering the set of $k$-ary strings of length $n$ (so we can generalize in later sections), let 

\begin{itemize}
    \item $A(w)$ denote the number of strings  with weight $\leq w$,
    \item $B(w,p)$  denote  the number of strings  with weight $w$, and period $p$, and
    \item $C(w,p)$ denote  the number of strings with weight $w$, and period $\leq p$.
\end{itemize}

\noindent
\noindent
In the binary case when $k=2$, clearly $A(w) = \sum_{j=0}^w {n \choose j}$. Recall from the observations in Example~\ref{exam:PCR} that $B(w,p) = pL_k(p,wp/n)$, assuming $p$ divides $n$.
Using these values, let

\begin{itemize}
    \item $m =$ the smallest weight $m$ such that $A(m) \geq L$,
    \item $h = $ the smallest period such that $A(m-1) + C(m,h) \geq L$, and
    \item $t = $ the smallest integer such that $A(m-1) +C(m,h-1) + th \geq L$.
\end{itemize}

\noindent
These values can be used to define the surplus $s$ as $A(m-1) + C(m,h-1) + ht - L$.
%
%

An MC is the result of joining together all PCR cycles of weight less than $m$ together with all PCR cycles of weight equal to $m$ and period less than $h$ together with exactly $t$ PCR cycles of weight $m$ and period $h$. As cycles are joined, a counter is maintained to keep track of the number of cycles of weight $m$ and period $h$ already joined into the MC.  Thus, the specific $t$ PCR cycles joined are not necessarily known \emph{a priori}. 
Etzion's original presentation adds cycles of weight $m$  starting with the largest period. We made one minor departure from this approach by adding cycles of weight $m$ starting from the smallest period, which handles a special case defined later.

\begin{exam}  \label{exam:MC} \small
Consider $L=46$ and $n=6$.  Since  
$A(3) = 42$ and $A(4) = 57$, we have $m=4$.
Since
$B(4,1) = B(4,2) = 0, \  B(4,3) = 3, \ B(4,4) = B(4,5) = 0,  \ B(4,6) = 12,$
we have $C(4,5) = 3$ and $C(4,6) = 15$.  Thus $h=6$.  Since $A(3) + C(4,5) + h = 42+3+6 =51$, we have $t=1$ and surplus $s=5$. Applying {\sc PCR3} as the underlying method for joining cycles, the following illustrates the construction of the MC starting with 000000.   
\medskip

    \begin{center}
    \resizebox{!}{3.5in}{\includegraphics{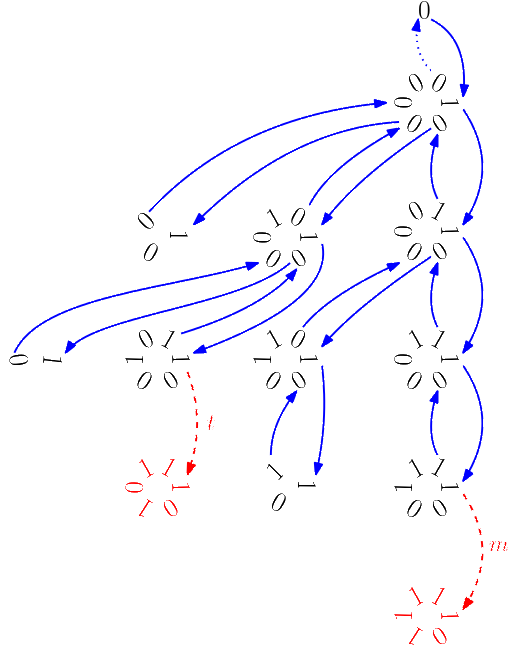}}
    \end{center}

 \noindent
The MC joins all cycles with weight less than $m=4$, 
all cycles with weight $m=4$ and period less than $h=6$, and exactly $t=1$ cycles with weight $m=4$ and period $h=6$.  Note that the cycle pointed to by the dashed arc labeled $m$ is not added since it has weight greater than $m=4$.
The cycle pointed to by the dashed arc labeled $t$ is not added since there was already $t=1$ cycles added with weight $m=4$ and period $h=6$.  The resulting MC with length 51 is 
$$
[~000000111100111000110110100110000101100101010001001~]. $$
\end{exam}

The second step involves \emph{cutting out} small cycles whose combined length totals the surplus $s$.  When $s=1$, the cycle $[0]$ is easily removed.  However, for arbitrary lengths, finding and removing such small cycles is non-trivial and the construction of the MC is critical for the ease in which these small cycles can be removed.  Etzion's approach requires cutting out up to $\lceil \log n \rceil$ cycles of the form $[0^{i}1]$ where $i+1<n$ is a power of 2. For example, when $n=14$, the possible cycles to remove would be of the form $[01]$, $[0001]$, $[00000001]$.  Depending on $s$, the cycle $[0]$ may also need to be cut out (the removal of a single 0).  Details on how to cut such cycles out are discussed in Section~\ref{sec:cut}.  There are two special cases that Etzion addresses to ensure that the aforementioned small cycles are indeed on the MC.  
\smallskip

\begin{itemize}
\item {\bf Special case \#1:} When $n=2m$, the cycle $[01]$ with weight $m$ must be included.
\item {\bf Special case \#2:} When $n=2m-1$ the cycle $[(01)^{m-1}1]$ with weight $m$ must be included.
\end{itemize}
\smallskip

As noted earlier, by considering the cycles of weight $m$ in increasing order by period, {\bf Special case \#1} is handled since the cycle $[01]$ is the unique cycle with period $2$ and will always be added in that case.  {\bf Special case \#2}  requires extra care when adding cycles of weight $m$ and period $n$, noting that there are clearly no cycles of weight $m$ and period less than $n$ when $n=2m-1$.

\section{Efficiently constructing binary cut-down DB sequences} \label{sec:binary}

In this section we apply the following two enhancements to simplify and improve Etzion's original approach for constructing a cut-down DB sequence. 

\vspace{-0.1in}

\begin{enumerate}
    \item We apply {\sc PCR3} and focus on a single cut-down DB sequence construction.
    \item We consider all cycles of the form $[0^i1]$ for $1\leq i \leq \lceil n/2 \rceil$ and cut out at most two small cycles.  
\end{enumerate}

\vspace{-0.1in}

\noindent
The latter step involves changing the definition of {\sc PCR3} for at most two strings.  
\edit{As noted in the previous section,  when defining an MC, the cycles of weight $m$ will be added by increasing (instead of decreasing) period, thus handling {\bf Special case \#1}.   To account for {\bf Special case \#2}, we assume that $\mathbf{S}$ contains all strings from the cycle $[(01)^{m-1}1]$  when $n=2m-1$.  Details for ensuring this are provided in the upcoming Algorithm~\ref{algo:cut}.}

We conclude this section by applying a ranking algorithm for fixed-weight Lyndon words to produce the first successor-rule construction of cut-down DB sequences.

\noindent

\subsection{Constructing a Main Cycle with PCR3$'$}
Recall that the substrings of length $n$ from an MC of length $L+s$ correspond to the set $\mathbf{S}$ of
\begin{itemize}
    \item all length-$n$ binary strings belonging to PCR cycles with weight less than $m$,
    \item all length-$n$ binary strings belonging to PCR cycles with weight $m$ and period less than $h$,  and
    \item a set $\mathbf{T}$ containing all length-$n$ binary strings from $t$ PCR cycles with weight $m$ and period $h$,
\end{itemize}
\noindent
where $m,h,t,s$ are the precomputed variables described in the previous section. 
%
Thus,  Algorithm~\ref{algo:MC} can be applied to construct a universal cycle for $\mathbf{S}$. 

For the remainder of this section, let $MC_\mathbf{S}$ denote the universal cycle for $\mathbf{S}$ obtained from Algorithm~\ref{algo:MC}.

\subsection{Cutting out small cycles} \label{sec:cut}

In order to cut down $MC_\mathbf{S}$ to a cycle of length $L$, ideally we  cut out a single substring of length $s$.  However, cutting out such a substring without introducing duplicate length-$n$ substrings is a challenge.  When $s \leq \lceil \frac{n}{2} \rceil$ we will demonstrate that finding such a substring is possible; otherwise we cut out two substrings whose combined length total $s$.   For our discussion let $F(a_1a_2\cdots a_n) = a_2\cdots a_nx$
where $x = $ {\sc PCR3}$'(a_1a_2\cdots a_n)$.

When $s=1$, it is straightforward to cut out the cycle $Z_1 = [0]$ by cutting the unique substring $z_1(1) = 0^n$ down to $0^{n-1}$.  
Otherwise, consider cycles of the form $Z_i= [0^{i-1}1]$ for $1 < i \leq \lceil\frac{n}{2} \rceil$.  Etzion~\cite{etzion86} considers similar cycles, but only those with length that is a power of two.  Let $\mathbf{Z}_i$ denote the set of $i$ length-$n$ strings belonging to $Z_i$.  We consider two cases for listing the $i$ strings in $\mathbf{Z}_i$ depending on whether or not $i$ divides $n$:
\begin{center}
\begin{tabular}{l|l}
~~~~~~~$i$ divides $n$ & ~~~~~$i$ does not divide $n$ \\ \hline
$z_i(1) = 0^{i-1}1\blue{(0^{i-1}1)^{a-1}}$      & $z_i(1) = 0^b1\blue{(0^{i-1}1)^{a}}$  \\
$z_i(2) = 0^{i-2}1~\blue{(0^{i-1}1)^{a-1}}~0$   & $z_i(2) =  0^{b-1}1~\blue{(0^{i-1}1)^{a}}~0$ \\
$z_i(3) = 0^{i-3}1~\blue{(0^{i-1}1)^{a-1}}~0^2$ & ~~~$\cdots$ \\
$~~~\cdots  $                             & $z_i(b{+}1)  =  1~\blue{(0^{i-1}1)^{a}}~0^b$ \\
$z_i(i)  =  1~\blue{(0^{i-1}1)^{a-1}}~0^{i-1}$  & $z_i(b{+}2) = 0^{i-1}1\red{(0^{i-1}1)^{a-1}}~0^{b+1}$ \\
                                                & $z_i(b{+}3)  = 0^{i-2}1\red{(0^{i-1}1)^{a-1}}~0^{b+2}$ \\
                                                & ~~~$\cdots$  \\
                                                & $z_i(i) = 0^{b+1}1\red{(0^{i-1}1)^{a-1}}~0^{i-1}$
\end{tabular}
\end{center}
where $a=\lfloor \frac{n}{i} \rfloor$ and $b = (n \bmod i)-1$. 
When $i$ divides $n$, the strings in $\mathbf{Z}_i$ belong to a single PCR cycle.   However, if $i$ does not divide $n$, then the strings in $\mathbf{Z}_i$ belong to two PCR cycles with different weights. 

\begin{exam}
Let $n=10$. Consider $i=2,3,4,5$, noting 2 and 5 divide 10: 

\begin{itemize}
    \item $\mathbf{Z}_2 = \{01\blue{01010101}, 10\blue{10101010}\}$,
    \item $\mathbf{Z}_3 = \{1\blue{001001001}, 001\red{001001}0, 01\red{001001}00\}$,
    \item $\mathbf{Z_4} = \{01\blue{00010001}, 1\blue{00010001}0, 0001\red{0001}00, 001\red{0001}000 \}$,
    \item $\mathbf{Z_5} = \{00001\blue{00001}, 0001\blue{00001}0, 001\blue{00001}00, 01\blue{00001}000, 1\blue{00001}0000 \}.$
\end{itemize}

\vspace{-0.1in}

\end{exam}
%

\noindent
 The reason why we cannot cut out a single cycle $[0^{s-1}1]$ when $s > \lceil n/2 \rceil$ is because the strings in $\mathbf{Z}_s$ do not appear contiguously on  $MC_\mathbf{S}$.

%
%

For $n \geq 2$, the strings in $\mathbf{Z}_i$ have weight at most $\lceil \frac{n}{2} \rceil$. This upper limit is obtained only when $i=2$ for $Z_2 =[01]$. Thus, the strings in $\mathbf{Z}_i$ always appear as substrings in $MC_\mathbf{S}$ since we already accounted for the two special cases defined at the end of Section~\ref{sec:etzion}.
%
The upcoming Lemma \ref{lem:F} demonstrates  that the strings $z_i(1), \ldots , z_i(i)$ appear contiguously as substrings on $MC_{\mathbf{S}}$.  Its proof applies the following obvious property of necklaces.
\begin{remark} \label{rem:neck}
A string $0^x1u0^{x+1}v$ is a {\bf not} a necklace for any integer $x$ and binary strings $u,v$.
\end{remark}

\begin{lemma}  \label{lem:F}
For $1 < i \leq \lceil \frac{n}{2} \rceil$ and $1 \leq j < i$, $F(z_i(j)) = z_i(j{+}1)$.
\end{lemma}
\begin{proof}
Let $z_i(j) = b_1b_2\cdots b_n$.
If $i$ divides $n$, then by Remark~\ref{rem:neck},  {\sc PCR3}$'$($z_i(j)) = b_1$ and the result holds.
If $i$ does not divide $n$, then again by Remark~\ref{rem:neck},  {\sc PCR3}$'$($z_i(j)) = b_1$ for all cases except $j=b+1$.  In this case $b_2\cdots b_n1$ is a necklace and therefore {\sc PCR3}$'$($z_i(j))  = \overline{b}_1$, and again the result holds.
\end{proof}

%
%
%

Algorithm~\ref{algo:step2} is obtained by making the following three modifications to Algorithm~\ref{algo:MC} for some $1 \leq i \leq \lceil \frac{n}{2} \rceil$:
\begin{enumerate}
    \item the initial string $\alpha$ is in $\mathbf{S} \setminus \mathbf{Z}_i$ (i.e, it is not one of the strings being cut out),
    \item the value for $x$ is complemented after Line 6 if $\beta = z_i(1)$ (cutting out the cycle $Z_i$), and
        \item the {\bf for} loop iterates $i$ less  times (to account for the cycle being cut out).
\end{enumerate}
\begin{algorithm}[th]           
\caption{Pseudocode for constructing a universal cycle for $\mathbf{S} \setminus \mathbf{Z}_i$, where $\alpha \in \mathbf{S} \setminus \mathbf{Z}_i$.
Assume $\mathbf{S}$ includes the strings from the cycle  $[(01)^{m-1}1]$  when $n=2m-1$.} \label{algo:step2}   
        
\begin{algorithmic} [1]                   
\Procedure{UC2}{$\alpha$}
\For{$i\gets 1$ {\bf to} $L+s-i$}
    \State \Call{Print}{$a_1$}

    \State $x \gets$ \Call{PCR3}{$\alpha$}
    \State $\beta \gets a_2\cdots a_n x$
     \If{$\beta \notin \mathbf{S}$ } \ \  $x \gets \overline{x}$ \EndIf
    \blue{ \If{$\beta = z_i(1)$} \ \ $x \gets \overline{x}$ \EndIf}
    \State $\alpha \gets a_2\cdots a_n x$ 
  
\EndFor

\EndProcedure
\end{algorithmic}
\end{algorithm} 

\begin{lemma} \label{thm:cut}
Algorithm~\ref{algo:step2} constructs a cut-down DB sequence of length $L+s-i$. 
\end{lemma}
\begin{proof}
Consider Algorithm~\ref{algo:MC}. It constructs a universal cycle $MC_\mathbf{S}$  of length $|\mathbf{S}| = L+s$.  Recall that $\mathbf{Z}_i \subseteq \mathbf{S}$. If the algorithm is initialized with a string in $\mathbf{S} \setminus \mathbf{Z}_i$,
then by Lemma~\ref{lem:F}, the first time $\alpha$ corresponds to a string in $\mathbf{Z}_i$ is when $\alpha = z_i(1)$.  Let $\alpha_0$ be the string such that $F(\alpha_0) = z_i(1)$.  Then

\medskip

$\alpha_0 = \left\{ \begin{array}{ll}
          10^{n-1}		&\ \  \mbox{if $i=1$;}\\
           0~\blue{(0^{i-1}1)^{a-1}}~0^{i-1}		&\ \  \mbox{if $i>1$ and $i$ divides $n$;}\\
          10^{b}1\red{(0^{i-1}1)^{a-1}}~0^{i-1}    &\ \   \mbox{if $i$ does not divide $n$,}
\end{array} \right.$

\medskip

\noindent
where $a=\lfloor \frac{n}{i} \rfloor$ and $b = (n \bmod i)-1$.
In each case $(\alpha_0, z_i(i))$ is a conjugate pair.  Thus, complementing the value for $x$ after Line 6 when $\beta = z_i(1)$ cuts out precisely the strings in $\mathbf{Z}_i$, effectively reversing the cycle joining detailed in Lemma~\ref{thm:concat}; i.e., it cuts out the cycle $Z_i$.  
Finally, by limiting the number of iterations of the {\bf for} loop to $L+s-i$ to account for cutting out this cycle, the resulting Algorithm~\ref{algo:step2} generates a cut-down DB sequence of length $L+s-i$.
\end{proof}

Let $j = \lceil n/2 \rceil$. If $s \leq j$, then let  $\mathbf{R} = \{ z_s(1) \}$; otherwise let $\mathbf{R} = \{ z_j(1), z_{s-j}(1) \}$. 
Since the strings in each $\mathbf{Z}_i$ are distinct, we can modify Line 7 of Algorithm~\ref{algo:step2} to  test if $\beta \in \mathbf{R}$ to remove either one or two small cycles; it follows from the proof of Lemma~\ref{thm:cut} that the resulting algorithm will generate a cut-down DB sequence of length $L$.  

Algorithm~\ref{algo:cut} applies this modification along with implementation details required to efficiently test if a string belongs to $\mathbf{S}$.
%
It starts with $\alpha=0^{n-1}1$, which is a string that does not belong to any $\mathbf{Z}_i$. Computing the weight and period of the current length-$n$ string $\alpha$ leads to an $ \mathcal{O}(n)$-time membership tester for $\mathbf{S}$.  In this implementation, the subset $\mathbf{T}$ of $\mathbf{S}$ is not known \emph{a priori}. Thus, we keep track of how many cycles of weight $m$ and period $h$ we have seen so far, adding them if we do not exceed $t$. This is maintained by the counter $t'$.  In the special case when $n=2m-1$, a flag is set to make sure the cycle $[(01)^{m-1}1]$ is included; the first string visited on this cycle is $(01)^{m-1}1$.  

\begin{theorem}
Algorithm~\ref{algo:cut} generates a binary cut-down DB sequence of length $L$ in $\mathcal{O}(n)$ time per symbol using $\mathcal{O}(n)$ space.
\end{theorem}


\begin{algorithm}[ht]      
\caption{Pseudocode for constructing a binary cut-down DB sequence of length $L$ assuming precomputed values $m,h,t$ and the set $\mathbf{R}$}  
\label{algo:cut}      

\begin{algorithmic} [1]                   
\Procedure{Cut-down}{}
\State $\alpha = a_1a_2\cdots a_n \gets 0^{n-1}1$
\State $t' \gets 0$
\If{$n=2m-1$} \ \ $flag \gets 1$  
\Else \ \ $flag \gets 0$
\EndIf

\Statex
\For{$i\gets 1$ {\bf to} $L$}
    \State \Call{Print}{$a_1$}

    \Statex  
    \State \blue{$\triangleright$ UC-successor for the Main Cycle}
    \State $w \gets$ weight of $\alpha$  
    \State $x \gets$ \Call{PCR3}{$\alpha$}
    \State $\beta \gets a_2\cdots a_n x$
     \If{$w = m$ {\bf and} $w-a_1+x = m+1$ } \ \  $x \gets \overline{x}$ \EndIf
     \If{$w = m-1$ {\bf and} $w-a_1+x = m$ }
        \If{ $\per(\beta) > h$}  \ \  $x \gets \overline{x}$ \EndIf
        \If{ $\per(\beta) = h$} 
             \If{$\beta = (01)^{m-1}1$} \ \ $\mathit{flag}\gets 0$ \EndIf
             \State \blue{$\triangleright$ Cut out excess cycles of weight $m$ and period $h$}
             \If{$t'=t$ {\bf or} ($t'+1 = t$ {\bf and} $\mathit{flag} = 1$)}  \ \  $x \gets \overline{x}$  \ 
             \Else  \ \ $t' \gets t'+1$
             \EndIf

        \EndIf
     \EndIf
    
    \Statex
    \If{ $\beta \in \mathbf{R}$ } \ \ $x \gets \overline{x}$  \ \  \ \     \blue{$\triangleright$ Cut out small cycle(s)}   \EndIf
    \State $\alpha \gets a_2\cdots a_n x$

\EndFor

\EndProcedure
\end{algorithmic}
\end{algorithm}

\begin{exam}  \label{exam:cut-out} \small
Recall the MC from Example~\ref{exam:MC} of length 51 where $L=46$ and $s=5$.  Setting $\mathbf{R} = \{001001, 010101\}$, the cycles $[001]$ and $[01]$ are cut out to obtain a cut-down DB sequence of length $L$.  
This is illustrated below where the dashed red arcs are not followed.

    \begin{center}
    \resizebox{!}{2.7in}{\includegraphics{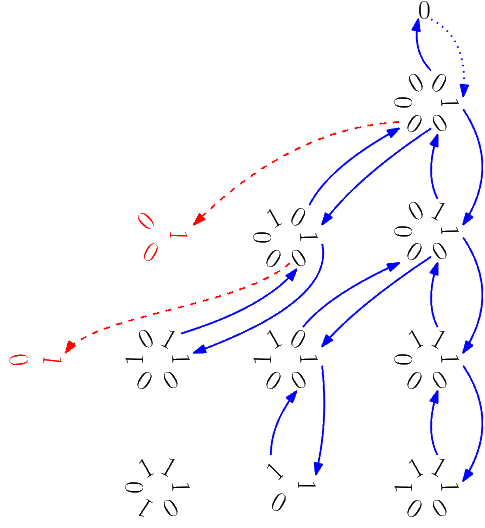}}
    \end{center}

\noindent
The resulting cut-down DB sequence of length $L=46$ starting from $\alpha = 000001$ is
$$
[~0000011110011100011011010011000010110010100010~].
$$

\noindent
{\bf Note:}
Cutting out cycles $Z_i$ where $i$ does not divide $n$ involves cutting chunks out of two PCR cycles.

\vspace{-0.05in}
\end{exam}

\subsection{A successor-rule construction} \label{sec:successor}

In this section we define a successor rule that can be used to construct a binary cut-down DB sequence of length $L$.  Unlike the algorithm in the previous section, no context is required when iterating through the successor rule; 
\edit{furthermore, determining whether or not a length-$n$ string is found on the sequence can be computed efficiently.}

Recall that the set $\mathbf{S}$, which contains the length-$n$ substrings of an MC, depends on a set $\mathbf{T}$ which contains 
$ht$ binary strings of length $n$, weight $m$, and period $h$. 
As there are many different possible sets $\mathbf{T}$ that meet this criteria, the key to our successor rule is to consider a single set $\mathbf{T}$. Specifically, let $\mathbf{T}$ denote the length-$n$ strings belonging to PCR cycles corresponding to the $t$ {\bf lexicographically largest} Lyndon words of length $h$ and weight $mh/n$.  Considering the $t$ largest Lyndon words instead of the $t$ smallest is important since it ensures the cycles required in the special cases are always included; each special cycle corresponds to the lexicographically largest Lyndon word for a given weight.  Let $\mathbf{S'}$ denote the set $\mathbf{S}$ that includes this specific subset $\mathbf{T}$.  

Determining whether or not a string belongs to $\mathbf{T}$ can be accomplished by applying a recent algorithm to rank fixed-weight Lyndon words as they appear in lexicographic order~\cite{hartman}.  Let $\alpha = a_1a_2\cdots a_n$ denote a binary string with period $n$ and weight $m$.  Let $\sigma$ be the lexicographically smallest rotation of $\alpha$, i.e., $\sigma$ is a Lyndon word.  Let {\sc RL}($a_1a_2\cdots a_n$) denote the rank of $\sigma$ in a lexicographic listing of length-$n$ Lyndon words with weight $m$. Computing $\sigma$ can be done in $\mathcal{O}(n)$ time~\cite{Booth} and computing the rank of $\sigma$ can be done with $\mathcal{O}(n^3)$ $n$-bit operations~\cite{hartman}.  Applying these functions, the procedure {\sc CutDownSuccessor}($\alpha$) given in Algorithm~\ref{algo:successor} will construct a cut-down DB sequence of length $L$; it is a universal cycle for $\mathbf{S'} \setminus \mathbf{C}$, where $\mathbf{C}$ consists of the strings cut out from the main cycle.  Specifically,  

$\mathbf{C} = \left\{ \begin{array}{ll}
          \mathbf{Z}_s		&\ \  \mbox{if $s \leq j$;}\\
          \mathbf{Z}_j \cup \mathbf{Z}_{s-j}			&\ \  \mbox{otherwise,}\\
\end{array} \right.$

\noindent
where $j= \lceil \frac{n}{2} \rceil$.
The algorithm can be initialized to any string $\alpha$ in $\mathbf{S'} \setminus \mathbf{C}$.

\begin{algorithm}[ht]           
\caption{A successor rule based construction of a cut-down DB sequence of length $L$ based on the precomputed values $m,h,t$ and the set $\mathbf{R}$.} \label{algo:successor}   
        
\begin{algorithmic} [1]                   
\Procedure{CutDownSuccessor}{$\alpha = a_1a_2\cdots a_n$}

\For{$i\gets 1$ {\bf to} $L$}
    \State \Call{Print}{$a_1$}
    
    \Statex

    \State \blue{$\triangleright$ Context-free UC-successor for the Main Cycle}
    \State $w \gets$ weight of $\alpha$ 
    \State $x \gets$ \Call{PCR3}{$\alpha$}
    \State $\beta \gets a_2\cdots a_nx$
    
    \If{$w > m$ {\bf or} ($w=m$ {\bf and} $\per(\beta) > h$) } \ $x \gets \overline{x}$ \EndIf
      
    \If{$w = m-1$ {\bf and} $w-a_1+x = m$ }
        \If{ $\per(\beta) > h$} \  $x \gets \overline{x}$ \EndIf
        \If{ $\per(\beta) = h$  {\bf and}  $L(h,mh/n) - $\Call{RL}{$a_1a_2\cdots a_h$} $+1 > t$ }    \  $x \gets \overline{x}$ \EndIf
     \EndIf
     \Statex
    \If{ $\beta \in \mathbf{R}$ } \ \ $x \gets \overline{x}$ \ \  \ \  \blue{$\triangleright$ Cut out small cycle(s)}   \EndIf   
    \State $\alpha \gets a_2\cdots a_n x$ 
  
\EndFor

\EndProcedure
\end{algorithmic}
\end{algorithm} 

\begin{theorem}
{\sc CutDownSuccessor}($\alpha$) generates a cut-down DB sequence of length $L$ starting from any $\alpha \in \mathbf{S'} \setminus \mathbf{C}$ where each symbol is generated using  $\mathcal{O}(n^{1.5})$-amortized simple operations on $n$-bit numbers and polynomial space.  \edit{Furthermore, determining whether or not a string $\beta$ belongs to $\mathbf{S'} \setminus \mathbf{C}$ can be computed using $\mathcal{O}(n^3)$ operations on $n$-bit numbers.}
\end{theorem}
\begin{proof}
By pre-computing $L_2(h,mh/n)$, each iteration of the {\bf for} loop 
requires $\mathcal{O}(n)$ time not including the time
required by a possible call to {\sc RL}.  A call to {\sc RL} is made exactly once for each cycle of weight $m$ and period $h$; it is made when
$\alpha$ has weight $m-1$ and $\beta$ has weight $m$ and period $h$.  These cycles correspond to Lyndon words of length $h$ and weight $mh/n$ and thus the call to {\sc RL} is made exactly $L_2(h,mh/n)$ times.  Clearly, $L_2(n,w) \leq {n \choose w} / {n}$,  and it is well-known that ${n \choose w}$ is $\mathcal{O}(2^n/\sqrt{n})$~\cite[p.~35]{Yudell1969}.  Since the running time of {\sc RL} is $\mathcal{O}(n^3)$, the work done by all calls to RL amortized over the $\Theta(2^n)$ iterations of the {\bf for} loop is $\mathcal{O}(\frac{n^3}{n\sqrt{n}}) = \mathcal{O}(n^{1.5})$.  The function {\sc RL} requires polynomial space to precompute some tables~\cite{hartman}.

\edit{
Consider a length-$n$ string $\beta$.
Since there are less than $n$ strings in $\mathbf{C}$, 
testing whether or not $\beta$ is in $\mathbf{C}$ can be determined in $\mathcal{O}(n^2)$ time.  
If $\beta$ has weight less than $m$, or if it has weight $m$ and period less than $h$, then it belongs to $\mathbf{S}$.  
If $\beta$ has weight greater than $m$, or weight equal to $m$ and period greater than $h$, then it does not belong to $\mathbf{S}$.  These cases can be handled in $\mathcal{O}(n)$ time. If
$\beta$ has weight $m$ and period $h$, then we compute {\sc RL}$(\beta)$ to determine if it is in $\mathbf{S}$ (see line 11 in Algorithm~\ref{algo:successor}); this takes $\mathcal{O}(n^3)$ time as noted earlier.  Thus, testing whether or not $\beta$ is in  $\mathbf{S'} \setminus \mathbf{C}$ can be determined in $\mathcal{O}(n^3)$ time.}
\end{proof}

\section{Cut-down DB sequences for $k>2$} \label{sec:kary}

In this section we extend the strategy for constructing binary cut-down DB sequences to alphabets of arbitrary size.\footnote{In \cite {etzion86}, Etzion concludes by stating that his binary construction of cut-down DB sequences can be generalized to alphabets of arbitrary size; however, no details are provided.} 
When generalizing the binary approach, the selection of an underlying de Bruijn successor is critical to a simple construction for a cut-down DB sequence when $k>2$.  The {\sc PCR3} successor applied in the binary case has two natural generalizations for $k \geq 2$. These generalizations have been previously 
defined as $g_3$ and $g'_3$ in~\cite{karyframework}.
The key to selecting an underlying de Bruijn successor is to allow for the simplest possible method to cut out small cycles; for our purposes, $g'_3$ is the one that allows for this. It can be computed in $\mathcal{O}(n)$ time using $\mathcal{O}(n)$ space~\cite{karyframework}; we relabel this de Bruijn successor to {\sc PCR3}$_k$ below\footnote{{\sc PCR3}$_k$ is labeled PCR3 (alt) in~\cite{dborg}.}, noting $\alpha = a_1a_2\cdots a_n$.

\begin{result}
\noindent {\bf {\sc PCR3}$_k$ de Bruijn successor}:

\medskip

\noindent
Let $y$ be the smallest symbol in $\{1,2,\ldots, k{-}1\}$  such that $a_2a_3\cdots a_{n} y$ is a necklace, or $y=0$ if no such symbol exists.
Then:

\smallskip

{\sc PCR3}$_k(\alpha) = \left\{ \begin{array}{ll}
           k{-}1		&\ \  \mbox{if  $y > 0$ and $a_1 = y{-}1$;}\\
          a_1{-}1		&\ \   \mbox{if  $y > 0$ and $a_1 > y{-}1$;}\\
         {a_1} \  		&\ \  \mbox{otherwise.}\end{array} \right.$

\end{result}

\noindent
Interestingly, the original presentation of {\sc PCR3}$_k$ is described as a UC-successor for $k$-ary strings with weight less than or equal to some fixed $w$.  However, we choose to apply the restriction separately from the successor as we additionally must consider the periods of the cycles.

The challenge when extending to larger alphabets is that the cycle-joining approach may no longer apply disjoint conjugate pairs. Instead, several cycles which have common substrings of length $n{-}1$ can be joined in a cyclic fashion; the same string can belong to more than one conjugate pair used during the cycle-joining process.
As an example, see Figure~\ref{fig:T43} which illustrates  how {\sc PCR3}$_k$ joins PCR cycles for $n=3$ and $k=4$ to obtain the DB sequence
 \[ [~0\blue{003303203103}002302202102001301201133132131123122333232221211101~].\]  
As a specific example of how {\sc PCR3}$_k$ joins cycles, consider the cycles $[003]$, $[031]$, $[032]$, and $[033]$.  Consider a UC that has joined $[003]$ but none of the other listed cycles.  Then $[031]$ is joined via the conjugate pair (003, 103).  Subsequently, $[032]$ is joined via the conjugate pair (003, 203), and finally $[033]$ is joined via the conjugate pair (003, 303).  Observe that 003 is used in all three conjugate pairs.  This leads to the substring
\blue{003303203103} highlighted in the above DB sequence obtained via repeated application of Lemma~\ref{thm:concat}.  Starting from 003, the cycles are visited in the order [003], [033], [032], [031] as illustrated in Figure~\ref{fig:T43}.

\begin{figure}[ht]
    \centering
    \resizebox{!}{5.1in}{\includegraphics{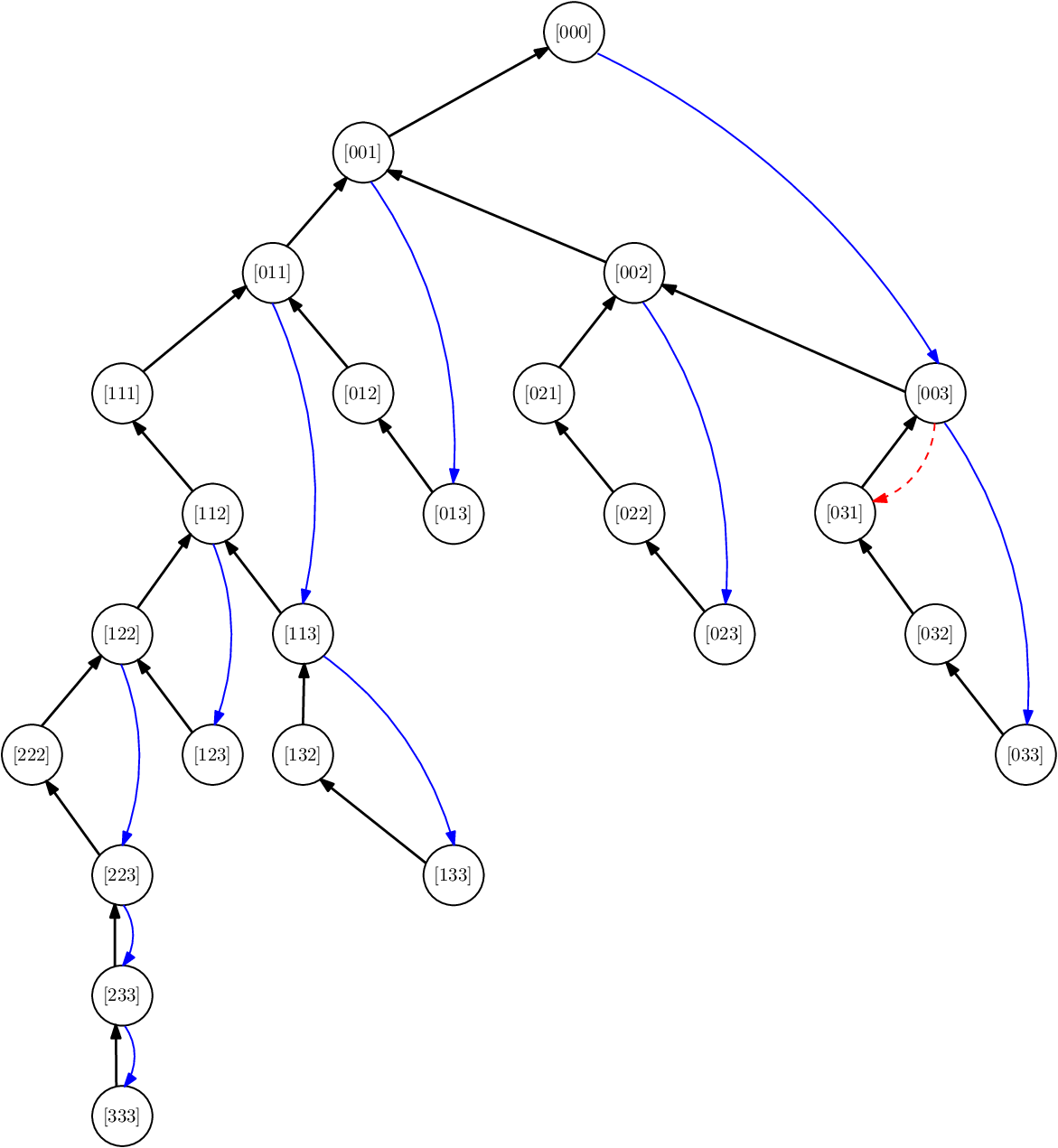}}
    \caption{Joining PCR cycles for $n=3$, $k=4$ by applying {\sc PCR3}$_k$.  The cycles drawn at the same level have the same weight. When constructing an MC with maximum weight $m = 4$, the dashed red edge illustrates how the cycles [032] and [033] can be cut out while still including [031]. The directions of the arrows indicate the order the cycles are visited when starting from [000]. }
    \label{fig:T43}
\end{figure}


\subsection{The MC for $k>2$}
As in the binary case, let $\mathbf{S}$ denote the set of strings belonging to an MC with a corresponding subset $\mathbf{T}$ of $ht$ strings belonging to $t$ cycles of weight $m$ and period $h$. Recall $|\mathbf{S}| = L + s$ and $k^{n-1} < L \leq k^n$.   The choice of the underlying de Bruijn successor {\sc PCR3}$_k$ allows for a simple construction of the MC when $k>2$. 

Given $\alpha = a_1a_2\cdots a_n$, let 
$F_k(\alpha) = a_2\cdots a_n${\sc PCR3}$_k(\alpha)$.
Like the binary case, we need only make a minor modification to the de Bruijn successor when attempting to branch to a PCR cycle with larger weight that does not belong to the MC.  In particular, if $F_k(\alpha)$ is not in $\mathbf{S}$, then it must be that $y=a_1+1$ is the smallest value such that $a_2\cdots a_ny$ is a necklace. This is the only case where $F_k(\alpha)$ has weight greater than $\alpha$, noting
{\sc PCR3}$_k(\alpha) = k-1$.  Instead, the next symbol in the universal cycle should be the {\bf largest symbol so the resulting string belongs to ${\mathbf S}$}.  This ensures we only cut out PCR cycles containing strings that do not belong to $\mathbf{S}$.   The resulting algorithm is detailed in Algorithm~\ref{algo:MC2}.  It differs from 
Algorithm~\ref{algo:MC} only at Line 6 to handle when attempting to branch to a PCR cycle not in the MC.
%


\begin{algorithm}[th]           
\caption{Pseudocode for constructing a universal cycle for $\mathbf{S}$ assuming the input $\alpha = a_1a_2\cdots a_n \in \mathbf{S}$ and $k>2$.} \label{algo:MC2}   
        
\begin{algorithmic} [1]                   
\Procedure{$k$-MC}{$\alpha$}
\For{$i\gets 1$ {\bf to} $|\mathbf{S}|$}
    \State \Call{Print}{$a_1$}

    \State $x \gets ${\sc PCR3}$_k(\alpha)$
    \State $\beta \gets a_2\cdots a_n x$
     \If{$\beta \notin \mathbf{S}$ } \ \  $x \gets$ the largest symbol such that $a_2\cdots a_n x \in \mathbf{S}$ \EndIf

    \State $\alpha \gets a_2\cdots a_n x$ 
  
\EndFor

\EndProcedure
\end{algorithmic}
\end{algorithm} 

\begin{lemma}
Algorithm~\ref{algo:MC2} generates a universal cycle for $\mathbf{S}$.
\end{lemma}
%

\begin{exam}
Consider Figure~\ref{fig:T43} where $m=4$ and $h=3$.  Consider $\alpha = 003$, which belongs to the cycle [003].  Observe that $F_k(\alpha) = 033$. 
This string belongs to [033] with weight 6 and hence is not on the MC.  If we stay on the cycle 003 by setting $x=0$, then we also \emph{cut out} cycles [032] and [031] which may contain strings on the MC.  However, if we have not already visited $t$ cycles with weight $m=4$, then we still
want to join the cycle [031], as illustrated by the red dashed edge in Figure~\ref{fig:T43}; i.e., $x$ should be assigned 1, which is the largest symbol such that $03x \in \mathbf{S}$.
\end{exam}

Using the same notation as the binary case, let $MC_\mathbf{S}$ denote the universal cycle for $\mathbf{S}$ obtained from Algorithm~\ref{algo:MC2}.  Next, we demonstrate how small cycles $Z_i$ can be cut out of $MC_\mathbf{S}$ to obtain a universal cycle of the desired length $L$.

\subsection{Cutting out small cycles for $k>2$}

Like the binary case, Algorithm~\ref{algo:MC2} can be applied to construct a specific MC  by counting the number of cycles added to the MC of weight $m$ and period $h$.  By the choice of the successor {\sc PCR3}$_k$, we can cut out the same small cycles $Z_i$ using the same set $\mathbf{R}$ from the binary case.  The resulting Algorithm~\ref{algo:cut2} will construct a cut-down DB sequence of length $L$ for $k>2$.  The key differences from the binary algorithm are as follows:
\begin{itemize}
    \item The initial string is generalized to $0^{n-1}(k{-}1)$, as it is the substring following $0^n$ in the DB sequence generated by {\sc PCR3}$_k$.
    \item The two special cases no longer need to be considered since $m$ will always be greater than $\lceil n/2 \rceil$ with $k>2$.  Thus, the variable \emph{flag} is no longer required.
    \item Lines 9-15 apply simple operations to cut out appropriate cycles based on the definition of {\sc PCR3}$_k$ to obtain the desired MC. 
    \item Line 17 assigns $x$ to 0, which is the same as complementing $x$ when $x=1$ in the binary case.  The strings in $\mathbf{R}$ remain the same.
\end{itemize}


\begin{algorithm}[ht]      
\caption{Pseudocode for constructing a cut-down DB sequence (for $k>2$) of length $L$ assuming precomputed values $m,h,t$ and the set $\mathbf{R}$.}  
\label{algo:cut2}      

\begin{algorithmic} [1]                   
\Procedure{$k$-Cut-down}{}
\State $\alpha = a_1a_2\cdots a_n \gets 0^{n-1}(k{-}1)$
\State $t' \gets 0$

\For{$i\gets 1$ {\bf to} $L$}
    \State \Call{Print}{$a_1$}
    
    \Statex  
    \State \blue{$\triangleright$ UC-successor for the Main Cycle}
    \State $w \gets$ weight of $\alpha$  
    \State $x \gets$ {\sc PCR3}$_k(\alpha)$
     \If{$w-a_1+x \geq m$ }  
        \State $x \gets m-w+a_1$ 
        \State $\beta \gets a_2\cdots a_n x$
        \If{ $\per(\beta) > h$}  \ \  $x \gets x-1$   \ \ \blue{$\triangleright$ Cut out cycles of weight $m$ and period $ > h$}  \EndIf
        \If{ $\per(\beta) = h$} 
            
             \If{$t'=t$} \ \  $x \gets x-1$  \ \  \blue{$\triangleright$ Cut out excess cycles of weight $m$ and period $h$} 
             \Else  \ \ $t' \gets t'+1$
             \EndIf

        \EndIf
     \EndIf
    
    \Statex
    \State $\beta \gets a_2\cdots a_n x$
    \If{ $\beta \in \mathbf{R}$ } \ \ $x \gets 0$  \ \  \ \     \blue{$\triangleright$ Cut out small cycle(s)}   \EndIf
    \State $\alpha \gets a_2\cdots a_n x$ 
 
\EndFor

\EndProcedure
\end{algorithmic}
\end{algorithm}

\begin{theorem} \label{thm:cut2}
Algorithm~\ref{algo:cut2} generates a cut-down DB sequence of 
 length $k^{n-1} < L \leq k^n$ where $k>2$ requiring $\mathcal{O}(n)$ time per symbol and using $\mathcal{O}(n)$ space.
\end{theorem}

\begin{proof} 
The modifications from Algorithm~\ref{algo:MC2} to include only PCR cycles in $MC_{\mathbf{S}}$ are straightforward (Lines 9-15).  Thus, we focus on Line 17, which cuts out one or two cycles of the form $Z_i$, depending on $\mathbf{R}$.  Since the strings in these cycles are disjoint, we focus 
on a single case of $z_i(1) \in \mathbf{R}$.
Since $k>2$, each string in $\mathbf{Z}_i$ is clearly a substring of $MC_\mathbf{S}$; they have small weight. By applying the same arguments from
the proof of Lemma~\ref{lem:F}, we obtain similar results: $F_k(z_i(j)) = z_i(j+1)$ for $1 < i \leq \lceil \frac{n}{2} \rceil$.  Let $\alpha_0$ be the string such that $F_k(\alpha_0) = z_i(1)$; there is a minor difference from the binary case when $i$ divides $n$:

\medskip

$\alpha_0 = \left\{ \begin{array}{ll}
          10^{n-1}		&\ \  \mbox{if $i=1$;}\\
           \red{2}~(0^{i-1}1)^{a-1}~0^{i-1}		&\ \  \mbox{if $i>1$ and $i$ divides $n$;}\\
          10^{b}1(0^{i-1}1)^{a-1}~0^{i-1}    &\ \   \mbox{if $i$ does not divide $n$,}
\end{array} \right.$

\smallskip

\noindent
where $a=\lfloor \frac{n}{i} \rfloor$ and $b = (n \bmod i)-1$.
When $i=1$, {\sc PCR3}$_k(\alpha_0) = 0$ and for the latter two cases {\sc PCR3}$_k(\alpha_0) = 1$. Thus, indeed $F_k(\alpha_0) = z_i(1)$.  
In each case $(\alpha_0, z_i(i))$ is a conjugate pair and by changing the value $x$ to {\sc PCR3}$_k(z_i(i))$ we cut out the cycle $Z_i$ (effectively reversing an application of Lemma~\ref{thm:concat}).  When $i=1$, {\sc PCR3}$_k(z_1(1)) = k{-}1$; otherwise {\sc PCR3}$_k(z_i(i)) = 0$. 
Observe, however, that we do not need to consider the case when $i=1$, since the algorithm is initialized to $0^{n-1}(k{-}1)$ and
the string $\alpha = 10^{n-1}$ is visited in the final iteration of the {\bf for} loop. This is why we can simply assign $x \gets 0$ at Line 17.
By limiting the number of iterations of the for loop to $L$ to account for cutting out the one or two small cycles, the resulting Algorithm~\ref{algo:cut2} generates a cut-down DB sequence of length $L$.


\end{proof}

\subsection{Precomputing $m,h,t,s$ for $k>2$}

\edit{The time to precompute the values $m,h,t,s$ for $k >2$ depend on the time to enumerate $T_k(n',w)$ for all $0 \leq n' \leq n$ and $0 \leq w \leq (k-1)n$. These values can be computed in $\mathcal{O}(k^2n^2)$ time applying dynamic programming techniques to the following recurrence for $n \geq 0$:}

\smallskip
\edit{
$T_k(n',w) = \left\{ \begin{array}{ll}
           0		&\ \  \mbox{if $w<0$ or ($n'=0$ and $w>0$);}\\
           1		&\ \  \mbox{if $n'=0$ and $w=0$;}\\
           \sum_{j=0}^{k-1} ~T_k(n'-1,w-j)     &\ \   \mbox{otherwise.}\end{array} \right.$}

\begin{exam}
We illustrate the computations of $m,h,t,s$ for $n=6$, $k=3$ and $L=617$.  First, we  compute the following table of values for $T(n',w)$:
\begin{center}
\begin{tabular}{c | rrrrrrrrr}
     $n'~\backslash~w$ & 0 & 1 & 2 & 3 & 4 & 5 & 6 & 7 & 8 \\ \hline
1  &     1 &     1 &     1 &     0 &     0 &     0 &     0 &     0 &     0  \\ 
2  &     1 &     2 &     3 &     2 &     1 &     0 &     0 &     0 &     0  \\
3  &     1 &     3 &     6 &     7 &     6 &     3 &     1 &     0 &     0  \\
4  &     1 &     4 &    10 &    16 &    19 &    16 &    10 &     4 &     1  \\
5  &     1 &     5 &    15 &    30 &    45 &    51 &    45 &    30 &    15 \\ 
6  &     1 &     6 &    21 &    50 &    90 &   126 &   141 &   126 &    90  \\    
\end{tabular}
\end{center}
Recall the definitions of $A(w)$, $B(w,p)$, and $C(w,p)$ defined in Section~\ref{sec:etzion},
for $n=6$ and $k=3$.
Since $A(7) = 561$ and $A(8) = 651$, we have $m=8$.
Since $B(8,1) = B(8,2) = 0$, $B(8,3) = 6$, $B(8,4) = B(8,5) = 0$,  and $B(8,6) = 84$, we have
$C(8,5) = 6$ and $C(8,6) = 90$. 
Thus $h=6$.  Since $A(8) + C(8,5) + 9h = 621$, we have $t=9$ and surplus $s=4$.

\end{exam}

\section{Summary and future work}

In this paper, we have enhanced Etzion's algorithm~\cite{etzion} to construct binary cut-down DB sequences.  Moreover, we generalize the algorithm to alphabets
of arbitrary size by selecting an appropriate underlying feedback function. The resulting algorithms run in $\mathcal{O}(n)$-time per symbol using $\mathcal{O}(n)$ space after some initialization requiring polynomial time and space; they are available for download at~\url{http://debruijnsequence.org/db/cutdown}~\cite{dborg}.
By utilizing an efficient algorithm to rank fixed-weight Lyndon words, we developed the first successor-rule construction for binary cut-down DB sequences that only requires the current length-$n$ substring to determine the next bit. It requires  $\mathcal{O}(n^{1.5})$-amortized
simple operations on $n$-bit numbers per bit.  However, it is important to note that the efficient ranking algorithm only applies to the binary case.

It is not difficult to observe that the cut-down DB sequences produced by our algorithms are not \emph{balanced}.
  Thus,  avenues for future research include:
\begin{enumerate}

    \item Develop an efficient (cycle-joining) construction for generalized DB sequences.
    \item Develop an efficient (cycle-joining) construction for balanced cut-down DB sequences.
        \item Develop an efficient ranking algorithm for fixed-weight Lyndon words and necklaces for $k>2$.
\end{enumerate}

\bibliographystyle{acm.bst}

\bibliography{cutdown-db}

\end{document}